\pdfoutput=1
\documentclass{llncs}

\usepackage[lined,noend,linesnumbered]{algorithm2e}
\usepackage{amsfonts}
\usepackage{amsmath}
\usepackage{amssymb}
\usepackage{color}
\usepackage[utf8]{inputenc}
\usepackage{graphicx}
\usepackage{caption}
\usepackage{subcaption}
\usepackage{wrapfig}
\usepackage{xspace}
\usepackage{float}
\usepackage{subcaption}

\DeclareUnicodeCharacter{2229}{$\cap$}
\DeclareUnicodeCharacter{3BC}{\ensuremath{\mu}}
\DeclareUnicodeCharacter{B5}{\ensuremath{\mu}}
\DeclareUnicodeCharacter{3BD}{\ensuremath{\nu}}

\newfloat{algfloat}{thbp}{alg}

\usepackage[affil-it]{authblk}

\newcommand{\set}[1]{\left\{#1\right\}}
\newcommand{\given}{\ \middle|\ }
\newcommand{\pinf}{{+\infty}}

\definecolor{darkgreen}{rgb}{0,0.3,0}

\newcommand{\play}[1]{\langle#1\rangle}
\newcommand{\allied}[1]{#1{\ \bmod\ 2}}
\newcommand{\operation}[1]{{\bf #1}}
\newcommand{\attr}{{\operation{Justify}}}
\newcommand{\PPG}{{\PG_P}}
\newcommand{\PG}{\mathcal{PG}}
\newcommand{\defwin}{{H_d}}
\newcommand{\Owns}{O}
\newcommand{\oppon}{{\bar{\alpha}}}

\newcommand{\ignore}[1]{}
\newcommand{\ignoreproof}[1]{}
\newcommand{\comments}[1]{}

\newcommand{\algsetup}{
\SetInd{0.45em}{0.45em}
\LinesNumbered
\SetKwProg{Fn}{Fn}{:}{}
\SetKwInOut{empty}{}
\SetKwData{is}{$\leftarrow$}
\SetKwInOut{Input}{  \xspace input}
\SetKwInOut{State}{  \xspace state}
\SetKwInOut{Output}{output}
\SetKwRepeat{Repeat}{do}{until}
\SetKwRepeat{DoWhile}{do}{while}}

\title{Justifications and a Reconstruction of Parity Game Solving Algorithms}
\author{
 Ruben Lapauw \thanks{\email{ruben.lapauw@cs.kuleuven.be}, Supported by the IWT Vlaanderen} \and 
 Maurice Bruynooghe \thanks{\email{maurice.bruynooghe@cs.kuleuven.be}} \and
 Marc Denecker \thanks{\email{marc.denecker@cs.kuleuven.be}}\ }
\institute{KU Leuven, Dept. of Computer Science, B-3001 Leuven, Belgium}
\institute{}
\begin{document}
\maketitle
\begin{abstract}
Parity games are infinite two-player games played on directed graphs.
Parity game solvers are used in the domain of formal
verification. This paper defines parametrized parity games and
introduces an operation, \attr{}, that determines a winning strategy for
a single node. By carefully ordering \attr{} steps, we reconstruct
three algorithms well known from the literature.
\end{abstract}

\section{Introduction}
Parity games are games played on a directed graph without leaves by
two players, Even (0) and Odd (1). A node has an owner (a player) and an integer
priority. A play is an infinite path in the graph where the owner of a node
chooses which outgoing edge to follow.  A play and its nodes is
won by Even if the highest priority that occurs infinitely often is
even and by Odd otherwise. A parity game is solved when the winner of
every node is determined and proven.

Parity games are relevant for boolean equation
systems~\cite{mCRL2_10.1007/978-3-642-36742-7_15,DBLP:journals/corr/abs-1210-6414},
temporal logics such as LTL, CTL and
CTL*~\cite{DBLP:conf/dagstuhl/2001automata} and µ-calculus~\cite{DBLP:conf/stacs/Walukiewicz96,DBLP:conf/dagstuhl/2001automata}.
Many problems in these domains can be reduced to solving a parity game.
Quasi-polynomial time algorithm for solving them  exist
\cite{DBLP:conf/stoc/CaludeJKL017,DBLP:conf/spin/FearnleyJS0W17,DBLP:conf/mfcs/Parys19}.
However, all current state-of-the-art algorithms (Zielonka's
algorithm~\cite{DBLP:journals/tcs/Zielonka98},
strategy-improvement~\cite{DBLP:conf/csl/Schewe08},
priority
promotion~\cite{DBLP:journals/iandc/BenerecettiDM18,DBLP:conf/cav/BenerecettiDM16,DBLP:conf/hvc/BenerecettiDM16}
and tangle learning~\cite{DBLP:conf/cav/Dijk18}) are exponential.

We start the paper with a short description of the role of parity game
solvers in the domain of formal verification
(Section~\ref{sec:verif}).
In Section~\ref{sec:ppg}, we recall the essentials of parity games
and introduce parametrized parity games as a generalization of parity
games.
In Section~\ref{sec:just} we recall justifications, which we
introduced in \cite{DBLP:conf/vmcai/LapauwBD20} to store winning
strategies and to speed up algorithms.
Here we introduce safe justifications and define a \attr{} operation
and proof its properties.
Next, in Section~\ref{sec:recursive}, we reconstruct three algorithms
for solving parity games by defining different orderings over \attr{}
operations.  We conclude in Section~\ref{sec:conclusion}.
\section{Verification and parity game solving}\label{sec:verif}
Time logics such as LTL are used to express properties of interacting
systems. Synthesis consists of extracting an implementation with the
desired properties. Typically, formulas in such logics are handled by reduction
to other formalisms.  LTL can be reduced to
Büchi-automata~\cite{DBLP:conf/lics/VardiW86,DBLP:conf/cav/KestenMMP93}, determinized with
Safra's construction~\cite{DBLP:conf/focs/Safra88}, and transformed to
parity games~\cite{DBLP:conf/lics/Piterman06}.  Other modal logics
have similar reductions, CTL* can be reduced to
automata~\cite{DBLP:conf/cav/BernholtzVW94}, to
µ-calculus~\cite{DBLP:journals/tcs/CranenGR11}, and recently to
LTL-formulae~\cite{DBLP:journals/corr/abs-1711-10636}. All are
reducible to parity games.

One of the tools that support the synthesis of implementations for
such formulas is
Strix~\cite{DBLP:journals/acta/LuttenbergerMS20,DBLP:conf/cav/MeyerSL18},
one of the winners of the
SyntComp~2018~\cite{DBLP:journals/corr/abs-1904-07736} and
SyntComp~2019 competition. It reduces LTL formulas on the fly to
parity games. A game has three possible outcomes: (i) the parity game
needs further expansion,
(ii) the machine wins the game, i.e., an implementation is feasible,
(iii) the environment wins, i.e., no implementation exists.
Strix also extracts an implementation with the specified behaviour, e.g., as a Mealy machine. 

Consider a formula based on the well-known dining philosophers problem:
\begin{equation}
\label{eq:philo}
\begin{array}{lll}
G (\mathit{hungry}_A \Rightarrow F \mathit{eat}_A) \land && \text{If A is hungry, he will eventually eat}\\
G (\mathit{hungry}_B \Rightarrow F \mathit{eat}_B) \land && \text{If B is hungry, he will eventually eat}\\
G (\lnot eat_A \lor \lnot eat_B) && \text{A and B cannot eat at the same time.} \\
\end{array}
\end{equation}
Here $(G\phi)$ means $\phi$ holds in every future trace and $(F \phi)$
means $\phi$ holds in some future trace where a trace is a succession
of states.

Strix transforms the LTL-formula~\ref{eq:philo} to the parity game of
Figure~\ref{fig:philo}. The
machine (Even) plays in the square nodes and the environment (Odd) in
the diamond nodes.
By playing in state $b$ to $d$, and in state $f$ to $h$, Even wins every node as 2 is then the highest priority that occurs infinitely often in
every play. 
From the solution, Strix extracts a 2-state
Mealy machine (Figure~\ref{fig:mealy}). Its behaviour satisfies Formula~\ref{eq:philo}: both philosophers alternate eating regardless of their hunger.

\begin{figure}[tb]
\begin{minipage}{0.55\textwidth}
\centering{
\includegraphics[width=0.9\textwidth]{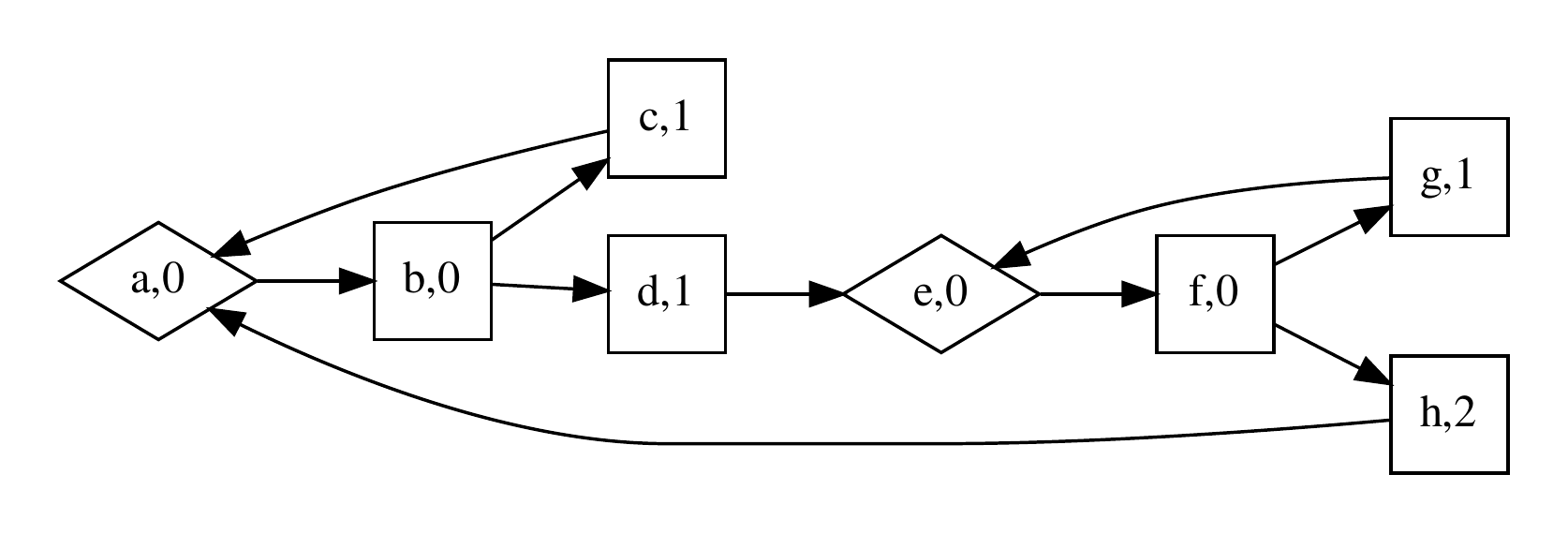}
\caption{A reduced parity game.\label{fig:philo}}
}
\end{minipage}\begin{minipage}{0.4\textwidth}
\centering{
\includegraphics[width=\textwidth]{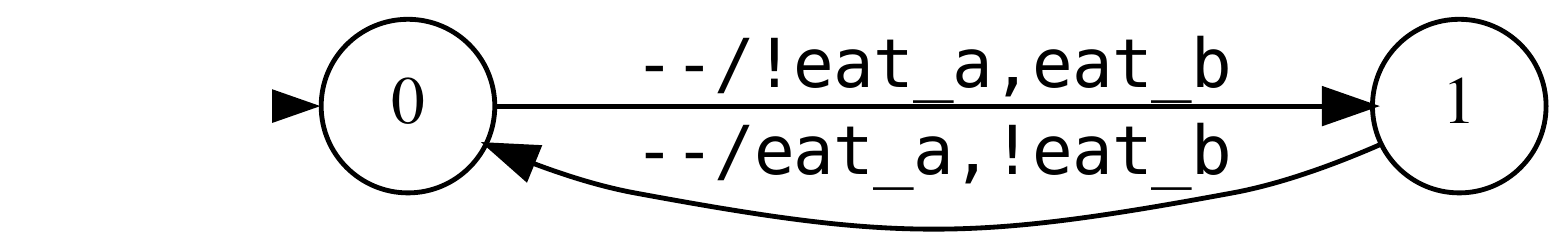}
\caption{The resulting Mealy machine with two
states, alternating $\neg eat_A, eat_B$ and $eat_A, \neg eat_B$ regardless of the input of $hungry_A$ and $hungry_B$.\label{fig:mealy}}
}
\end{minipage}
\end{figure}
\section{Parametrized parity games\label{sec:ppg}}

A parity game~\cite{mostowski1991games,DBLP:conf/focs/EmersonJ91,DBLP:conf/stacs/Walukiewicz96} is a two-player
game of player $0$ (Even) against $1$ (Odd). We use $\alpha \in
\set{0,1} $ to denote a player and $\oppon$ to denote its
opponent. Formally, we define a {\em parity game} as a tuple $\PG =
(V, E, \Owns, Pr)$ with $V$ the set of nodes, $E$ the set of possible
moves represented as pairs $(v,w)$ of nodes, $\Owns:V \to \{0,1\}$
the owner function, and $Pr$ the priority function $V \to \mathbb{N}$
mapping nodes to their priority; $(V,E)$ is also called the game
graph. Each $v\in V$ has at least one possible move. 
We use $\Owns_\alpha$ to denote nodes owned by $\alpha$.

A {\em play} (in node $v_1$) of the parity game  is an infinite
sequence of nodes $\play{v_1, v_2,\dots, v_n\dots}$ where $\forall i:
v_i \in V \land (v_i, v_{i+1}) \in E$. 
We use $\pi$ as a mathematical variable to denote a play. $\pi(i)$ is the $i$-th node $v_i$ of $\pi$.
In a play $\pi$, it is the owner of the node $v_i$ that decides the move $(v_i,v_{i+1})$.
There exists plays in every node.
We call the player $\alpha=(n\mod 2)$ the {\em winner of priority $n$}. 
The winner of a play is 
the winner of the highest priority $n$ through which the play passes infinitely often. Formally: 
$\mathit{Winner}(\pi) = \allied{\lim_{i \to \pinf} max\set{Pr(\pi(j)) \middle| j \geq i}}.$

The key questions for a parity game $\PG$ are, for each node $v$: Who is the winner? And how?
As proven by \cite{DBLP:conf/focs/EmersonJ91}, parity games are memoryless determined: every node has a unique winner and a corresponding memoryless winning strategy.
A (memoryless) strategy for player $\alpha$ is a partial function $\sigma_\alpha$ from a subset of $\Owns_\alpha$ to $V$. A play $\pi$ is consistent with $\sigma_\alpha$ if for every $v_i$ in $\pi$ belonging to the domain of $\sigma_\alpha$, $v_{i+1}$ is $\sigma_\alpha(v_i)$.
A strategy $\sigma_\alpha$ for player $\alpha$ is a {\em winning strategy} for a node $v$  if every play in $v$ consistent with this strategy is won by $\alpha$, i.e. regardless of the moves selected by $\oppon$.
As such, a game $\PG$ defines a winning function $W_\PG:V\mapsto
\{0,1\}$. 
The set $W_{\PG,\alpha}$ or, when $\PG$ is clear from the context,
$W_\alpha$ denotes the set of nodes won by $\alpha$.
Moreover, for both players $\alpha\in \{0,1\}$, there exists a memoryless winning  strategy $\sigma_\alpha$ with domain  $W_\alpha \cap \Owns_\alpha$ that wins in all nodes won by $\alpha$.  A {\em solution} of  $\PG$ consists of a function $W':V \to \{0,1\}$ and two winning strategies $\sigma_0$ and $\sigma_1$, with $dom(\sigma_\alpha)=W'_\alpha\cap \Owns_\alpha$, such that every play in $v\in W'_\alpha$ consistent with $\sigma_\alpha$ is won by $\alpha$. Solutions always exist; they may differ in strategy but all have $W'=W_{\PG}$, the winning function of the game. We can say that the pair $(\sigma_0,\sigma_1)$ proves that $W'=W_{\PG}$. 

In order to have a framework in which we can discuss different
algorithms from the literature, we define a parametrized parity
game. It consists of a parity game $\PG$ and a parameter function $P$,
a partial function $P:V\rightharpoonup \{0,1\}$ with domain $dom(P) \subseteq
V$. Elements of $dom(P)$ are called parameters, and $P$ assigns a
winner to each parameter. Plays are the same as in a $\PG$ except that
every play that reaches a parameter $v$ ends and is won by $P(v)$.

\begin{definition}[Parametrized parity game\label{def:ppg}]
  Let $\PG = (V, E, \Owns, Pr)$ be a parity game and $P:V\rightharpoonup \{0,1\}$ a partial function with domain $dom(P) \subseteq V$. Then $(\PG, P)$  is a {\em parametrized parity game} denoted $\PPG$,  with parameter set $dom(P)$. If $P(v)=\alpha$, we call $\alpha$ the assigned winner of parameter $v$. The sets $P_0$ and $P_1$ denote parameter nodes with assigned winner 0 respectively 1.  

A play of $(\PG, P)$ is a 
sequence of nodes $\play{v_0, v_1, \dots}$ such that for all $i$: if
$v_i \in P_\alpha$ then the play halts and is won by $\alpha$,
otherwise $v_{i+1}$ exists and $(v, v_{i+1}) \in E$. For infinite
plays, the winner is as in  the original parity game $\PG$. 
\end{definition}

Every parity game $\PG$ defines a class of parametrized parity games
(PPG's), one for each  partial function $P$. The original $\PG$
corresponds to one of these games, namely the one without parameters
($dom(P)=\emptyset$); every total function $P:V\to \{0,1\}$
defines a trivial PPG, with plays of length 0 and $P=W_\PPG$. 

A PPG $\PPG$ can be reduced to an equivalent PG $G$: in each parameter $v \in dom(P)$ replace the outgoing edges with a self-loop and the priority of $v$ with $P(v)$. We now
have a standard parity game $G$. Every infinite play $\play{v_0, v_1,
  \dots}$ in $\PPG$ is also an infinite play in $G$ with the same winner. 
Every finite play $\play{v_0, v_1, \dots, v_n}$ with winner $P(v_n)$ in $\PPG$
corresponds to an infinite play $\play{v_0, v_1, \dots, v_n,v_n,\ldots
}$ with winner $P(v_n)$ in $G$. Thus, the two games are equivalent. It follows that any PPG $\PPG$ is a zero-sum game defining a winning function $W$ and having memory-less winning strategies  $\sigma_\alpha$ with domain $(W_\alpha \setminus P_\alpha) \cap \Owns_\alpha$ (for $\alpha=0,1$).

PPG's allow us to capture the behaviour of several state of the
  art algorithms as a sequence of solved PPG's. In each step,
  strategies and parameters are modified and a solution for one PPG is
  transformed into a solution for a next PPG and this until a solution for the input
  PG is reached.

\begin{figure}[tb]
\begin{center}
\includegraphics[width=0.6\textwidth]{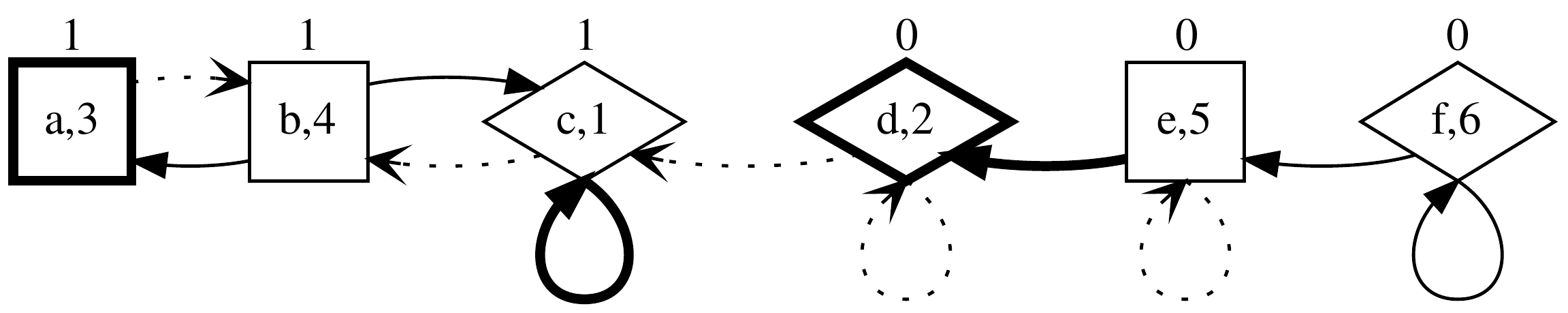}
\end{center}
\caption{A parametrized parity game with nodes $a,\dots,f$,  $P_0= \{d\}$ and $P_1 = \{a\}$, and winning strategies for $0$ and $1$. The two parameter nodes are in bold.  Square nodes are owned by Even, diamonds by
  	Odd.  The labels inside a node are the name  and priority; the label on top of a node is the winner. A bold edge belongs to a winning strategy (of the owner of its start node). A slim edge is one starting in a node that is lost by its owner. All remaining edges are dotted.\label{ex:game_para}
  }
\end{figure}

\begin{figure}[tb]
\begin{center}
\includegraphics[width=0.6\textwidth]{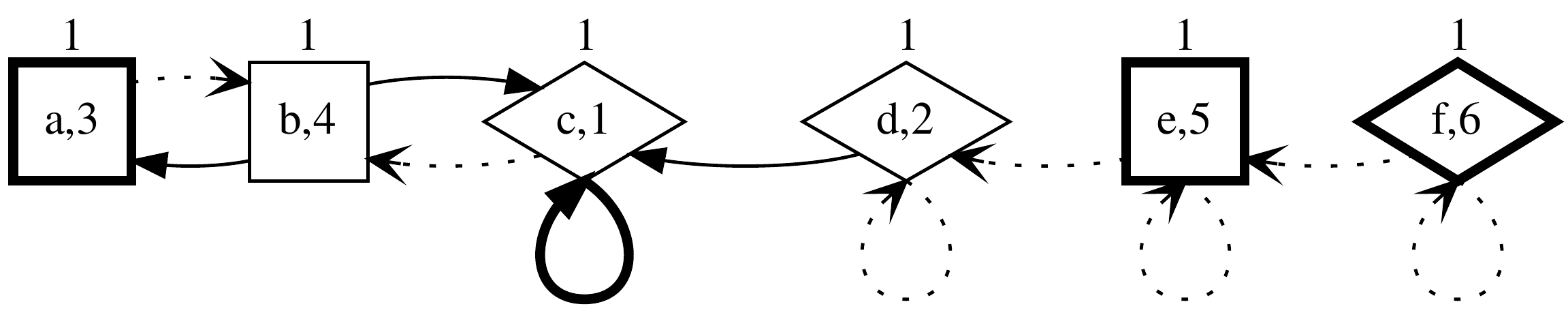}
\end{center}
\caption{{A parametrized parity game and strategy, after withdrawing $d$ from the parameter list.}
\label{ex:game_para_2}}
\end{figure}

\begin{example}\label{ex:first}
Figure~\ref{ex:game_para} shows a parametrized parity game and its
winning strategies.  The parameter nodes $a$ and $d$ are won by the
assigned winners, respectively 1 and 0.  Player 1 owns 
node $c$ and wins its priority. Hence, by playing
from $c$ to $c$, 1 wins in this node. Node $b$ is owned by 0 but
has only moves to nodes won by 1, hence it is also won by 1. Player 0
wins node $e$ by playing to node $d$; 1 plays in node $f$ but
playing to $f$ results in an infinite path won by 0, while playing
to node $e$ runs into a path won by 0, so $f$ is won by 0.

Based on this PPG, we can construct a solved PPG where node $d$ is removed from the parameters. The strategy is adjusted accordingly: Odd wins in $d$ by playing to $c$ . However, changing the winner of $d$ breaks the strategies and winners of the nodes $e$ and $f$. 
Figure~\ref{ex:game_para_2} shows one way to obtain a solved PPG with further adjustments: nodes $e$ and $f$ are turned into parameters won by $1$. Many other solutions exist, e.g., by turning $e$ into a parameter won by $0$.  
\end{example}
\newcommand{\dominion}{+\infty}
\newcommand{\ext}{{\operation{Extend}}}
\newcommand{\justify}{\ensuremath{\attr(J,v,dj)}}
\newcommand{\clo}{{\mathit{Close}}}
\newcommand{\attrt}{{\operation{Attract_{QD}}}}
\newcommand{\jl}{\ensuremath{jl}}
\newcommand{\DJN}{\mathit{DJ}}
\newcommand{\DJT}{\mathit{DJT}}
\newcommand{\djn}{{dj}}
\newcommand{\djt}{\djn}
\newcommand{\djl}{\mathit{djl}}
\newcommand{\defw}{\mathit{defwinner}}
\newcommand{\Ext}{\mathit{Par}}
\newcommand{\Esc}{\mathit{Esc}}

\newcommand{\reaching}[2]{#1{\downarrow_{#2}}}
\newcommand{\reachable}[2]{#1{\uparrow_{#2}}}

\section{Justifications\label{sec:just}}
In Figure~\ref{ex:game_para} and Figure~\ref{ex:game_para_2}, the
solid edges form the subgraph of the game graph that was analysed to
confirm the winners of all nodes.  We formalize this subgraph as a
{\em justification}, a concept introduced in
\cite{DBLP:journals/tplp/HouCD10} and described below.  In the rest of
the paper, we assume the existence of a parity game
$\PG=(V,E,\Owns,Pr)$ and a parametrized parity game $\PG_P=(\PG,P)$
with $P$ a parameter function with set of parameters $dom(P)$. Also,
we use $H: V \to \{0,1\}$ as a function describing a
``hypothesis'' of who is winning in the nodes.

\begin{definition}[Direct justification]\label{def:just_win}
  A {\em direct justification} $dj$ for player $\alpha$ to win node $v$ is a
  set containing one outgoing edge of $v$ if $\Owns(v)=\alpha$ and all
  outgoing edges of $v$ if $\Owns(v)= \oppon$.

  A direct justification $dj$ {\em wins $v$ for $\alpha$ under
    hypothesis $H$} if for all $(v,w)\in dj$, $H(w)=\alpha$.  We also
  say: {\em $\alpha$ wins $v$ by $dj$ under $H$}.
  \end{definition}

\begin{definition}[Justification]
  A  {\em justification} $J$ for $\PG$ is a tuple $(V,D,H)$ such that $(V,D)$ is a subgraph of $(V,E)$.
  If a node has outgoing edges in $D$, it is {\em justified} in $J$, otherwise it is {\em unjustified}.
\end{definition}

\begin{definition}[Weakly winning]
  A justification $(V,D,H)$ is {\em weakly winning} if for all justified nodes $v\in V$ the set of outgoing edges $Out_v$ is a direct justification that wins $v$ for $H(v)$ under $H$.
\end{definition} 

We observe that any justification $J = (V,D,H)$ determines a PPG  $\PG_{P_J}$ where the parameter function $P_J$ is the restriction of $H$ to unjustified nodes.

If $J$ is weakly winning, the set of edges $\{(v,w)\in D\mid
\Owns(v)=H(v)=\alpha\}$ is a partial function on $\Owns_\alpha$, i.e.,
a strategy for $\alpha$. We denote it as $\sigma_{J,\alpha}$.

\begin{proposition}\label{prop:weaklywinning}
	Assume  a weakly winning justification $J = (V,D,H)$. Then, (i) For every path $\pi$ in $D$, all nodes $v$ on $\pi$ have the same hypothetical winner $H(v)$. (ii) All finite paths $\pi$ starting in node $v$ in $D$ are won in $\PG_{P_J}$ by $H(v)$.  (iii) Every path in $D$ with nodes hypothetically won by $\alpha$ is consistent with $\sigma_{J,\alpha}$. (iv) Every play starting in $v$ of $\PG_{P_J}$  consistent with $\sigma_{J,H(v)}$ is a path in $D$.
\end{proposition}
\begin{proof}
 (i) Since any edge $(v,w)\in D$ belongs to a direct justification that wins $v$ for $H(v)$, it holds that $H(v)=H(w)$. It follows that every path $\pi$ in $D$ consists of nodes with the same hypothetical winner. (ii) If path  $\pi$ in $v$ is finite and ends in parameter $w$, then $H(v)=H(w)$. The winner of $\pi$ in $\PG_{P_J}$ is $P_J(w)$ which is equal to $H(v)$ as $H$ expands $P_J$. (iii) Every path in $D$ with hypothetical winner $\alpha$, follows $\sigma_{J,\alpha}$ when it is in a node $v$ with owner $\alpha$. (iv) Let $H(v)=\alpha$ and $\pi$ be a play in $v$ of $\PPG$  consistent with $\sigma_{J,\alpha}$. We can inductively construct a path from $v=v_1$ in $D$. It follows from (i) that the $n$'th node $v_n$ has $H(v_n)=H(v_1)=\alpha$. For each non-parameter node $v_n$, if $\Owns(v_n)=\alpha$, then $v_{i+1}=\sigma_{J,\alpha}(v_i)$ which is in $D$. If $\Owns(v_n)=\oppon$ then $D$  contains all outgoing edges from $v_n$ including the one  to  $v_{n+1}$. \qed
 \end{proof}

\begin{definition}[Winning]
  A justification $J = (V,D,H)$ is {\em winning} if (i) $J$ is weakly winning and (ii) all
  infinite paths $\langle v_1,v_2,\dots\rangle$ in $D$ are plays of
  $\PG$ won by $H(v_1)$.
\end{definition} 

Observe that, if $J$ is winning and $H(v)=\alpha$, all plays in $\PG_{P_J}$ starting in  $v$ and consistent with $\sigma_{(V,D,H),\alpha}$ are paths in
$(V,D)$ won by $\alpha$. Hence:
\begin{theorem}
  If $J=(V,D,H)$ is a winning justification for $\PG_{P_J}$ then $H$ is $W_{\PG_{P_J}}$, the
  winning function of $\PG_{P_J}$, with corresponding winning strategies $\sigma_{J, 0}$ and $\sigma_{J, 1}$.
\end{theorem}
The central  invariant of the algorithm presented below is that its data structure $J=(V,D,H)$ is a winning justification. Thus, in every stage, $H$ is the winning function of $\PG_{P_J}$ and the graph $(V,D)$ comprises winning strategies $\sigma_{J,\alpha}$ for both players. In a sense, $(V,D)$ provides a proof that $H$ is $W_{\PG_{P_J}}$. 

\subsection{Operations on weakly winning justifications}
We introduce an operation that modifies a justification $J=(V,D,H)$
and hence also the underlying game $\PG_{P_J}$. Let $v$ be a node in
$V$, $\alpha$ a player and $dj$ either the empty set or a direct
justification. We define $J[v:dj,\alpha]$ as the justification
$J'=(V,D',H')$ where $D'$ is obtained from $D$ by replacing the
outgoing edges of $v$ by the edges in $dj$, and $H'$ is the function
obtained from $H$ by setting $H'(v):=\alpha$.
Modifications for a set of nodes are independent of application
order. E.g., $J[v:\emptyset,H'(v)\mid v\in S]$  removes all
out-going edges of $v$ and sets $H'(v)$ for all $v\in S$. 
Multiple operations, like $J[v:dj,\alpha][v':dj',\alpha']$, are applied left to right.
Some useful instances, with their properties, are below.

In the proposition, a cycle in $J$ is a finite sequence of nodes following edges in $J$ that ends in its starting node. 

\begin{proposition}   \label{prop:operations}
For a weakly winning justification $J$ and a node $v$ with direct justification $dj$ the following holds:

(i) If $H(v)=\oppon$, $v$ has no incoming edges and $dj$ wins $v$ for $\alpha$ under $H$, then  $J[v:dj,\alpha]$ is weakly winning and there are no cycles in $J'$ with edges of $dj$. 
	 
	 (ii) Let $S$ be a set of nodes closed under incoming edges (if $v\in S$ and $(w,v)\in D$, then $w\in S$), let $H_f$ be an arbitrary function mapping nodes of $S$ to players. It holds that  $J[v:\emptyset,H_f(v) \mid v\in S]$ is weakly winning. There are no cycles in $J'$ with edges of $dj$.
	 
	 (iii) If $H(v)=\alpha$ and $dj$ wins $v$ for $\alpha$ under
  $H$,  then  $J[v:dj,\alpha]$ is weakly winning. There are no new
  cycles when $(v,v)\not\in dj$ and no $w\in range(dj)$ can reach $v$ in
  $J$. Otherwise new cycles pass through $v$ and have at least one
  edge in $dj$.
\end{proposition}
\begin{proof}  
  We exploit the fact that $J$ and $J'$ are very similar. 	

	 (i)  The direct justification $dj$ cannot have an edge ending in $v$ since $H(v)\neq H(w)$ for $(v,w)\in dj$ and no $w\in dj$ can reach $v$ in $J$ since $v$ has no incoming edges, hence $J'$ has no cycles through $dj$. As $J$ is weakly winning and $H$ is updated only in $v$, the direct justification of a justified node $w\neq v$ in $J$  is still winning in $J'$. Since also $dj$ wins $v$ for $\alpha$, $J'$ is weakly winning.     

  (ii) Setting $H(v)$ arbitrary cannot endanger the weak
  support of $J'$ as $v$ has no direct justification and no incoming
  edges in $J'$. Hence $J'$ is weakly winning. Also, removing direct
  justifications cannot introduce new cycles.
  
  (iii) Let $H(v)=\alpha$ and $dj$ wins $v$ for $\alpha$ under $H$.
  Let $J'=J[v:dj,\alpha]$. We have $H'=H$ so the direct justifications of all nodes
    $w \neq v$ in $J'$ win $w$ for $H'(w)$.  Since $dj$ wins $v$ for $H'(v)$,
    $J'$ is weakly winning.  Also, new cycles if any, pass through
  $dj$ and $v$.
\end{proof}

\subsection{Constructing winning justifications}
The eventual goal of a justification is to create a winning justification without unjustified nodes.
Such a justification contains a solution for the parity game without parameters.
To reach this goal we start with an empty winning justification and iteratively assign a direct justification to one of the nodes.

However, haphazardly (re)assigning direct justifications will violate the intended winning justification invariant.
Three problems appear:
First, changing the hypothesis of a node may violate weakly winning for incoming edges.
The easiest fix is to remove the direct justification of nodes with edges to this node.
Yet removing direct justifications decreases the justification progress.
Thus a second problem is ensuring progress and termination despite these removals.
Third, newly created cycles must be winning for the hypothesis.
To solve these problems, we introduce safe justifications; we start with
some auxiliary concepts.

Let $J$ be a justification.
The set of nodes {\em reaching} $v$ in J, including $v$, is closed
under incoming edges and is denoted with $\reaching{J}{v}$. 
The set of nodes {\em reachable} from $v$ in $J$, including $v$, is denoted with
$\reachable{J}{v}$.
We define $\Ext_J(v)$ as the parameters reachable from the node $v$,
formally $\Ext_J(v) = \reachable{J}{v} \cap dom(P)$.
The {\em justification level} $\jl_J(v)$ of a node $v$ is the lowest
priority of all its parameters and $\dominion$ if $v$ has none.
The {\em justification leve}l $jl_J(dj)$ of a direct justification
$dj=\{(v,w_1),\ldots,(v,w_n)\}$ is $min \{\jl_J(w_1), \ldots, \jl_J(w_n)\}$,
the minimum of the justification levels of the $w_i$.
We drop the subscript $J$ when it is clear from the context and write
$\Ext(v)$, $\jl(v)$ and $\jl(dj)$ for the above concepts.  The {\em
  default winner} of a node $v$ is the winner of its priority, i.e.,
$\allied {Pr(v)} $; the {\em default hypothesis} $H_d$ assigns default
winners to all nodes, i.e., $H_d(v)=\allied {Pr(v)}$.

\begin{definition}[Safe justification\label{def:extdom}]
A justification is safe iff (i) it is a winning justification, (ii)
all unjustified nodes $v$ have $H(v)=\defwin(v)$, that is, the winners of the current parameters of the PPG are their default winners, and (iii)
$\forall v \in V : \jl(v) \geq Pr(v)$, i.e., the justification level of
a node is at least its priority.
\end{definition}

Fixing the invariants is easier for safe justifications.  Indeed, for
nodes $w$ on a path to a parameter $v$, $Pr(v) \geq jl(w) \geq Pr(w)$,
so when $v$ is given a direct justification to $w$ then $Pr(v)$ is the
highest priority in the created cycle and $H(v)$ correctly denotes its
winner.
Furthermore, the empty safe justification $(V,\emptyset,\defwin)$ will serve as initialisation of the solving process.

\subsection{The operation Justify} 
To progress towards a solution, we introduce a single operation, namely
$\attr$. Given appropriate inputs, it can assign a direct justification
to an unjustified node or replace the direct justification of a
justified node. Furthermore, if needed, it manipulates the justification in order
to restore its safety. 

\begin{definition}[\attr] \label{def:justify}
  The operation \justify{} is {\em executable} if
  \begin{itemize}
  \item Precondition 1: $J = (V,D,H)$ is a safe justification, $v$ is
    a node in $V$, there exists a player $\alpha$ who wins $v$ by $dj$
    under $H$.
  \item Precondition 2: if $v$ is unjustified in $J$ then $\jl(dj)
    \geq \jl(v)$
    else
    $\jl(dj) > \jl(v)$. \\
  \end{itemize}
  Let \justify{} be executable.  If $H(v) = \alpha$ then $\justify{} = J[v: dj,H(v)]$, i.e., $dj$ becomes the direct justification of $v$.
  
  If $H(v) = \oppon$, then $\justify{} = {J[w:\emptyset,\defwin(w)\mid
      w\in \reaching{J}{v}][v:dj,\alpha]}$, i.e., $\alpha$ wins $v$ by
  $dj$, while all other nodes $w$ that can reach $v$ become unjustified, and
  their hypothetical winner $H(w)$ is reset to their default winner.
\end{definition}

If \justify{} is executable, we say that $v$ is {\em justifiable with $dj$} or {\em justifiable} for short;
when performing the operation, we {\em justify} $v$.

Observe, when \attr\ modifies the hypothetical winner $H(v)$, then, to preserve weak winning, edges $(w,v)$ need to be removed,  which is achieved by removing the direct justification of $w$. Moreover, to preserve (iii) of safety, this process must be iterated until fixpoint and to preserve (ii) of safety, the hypothetical winner $H(w)$ of $w$ needs to be reset to its default winner. This produces a situation satisfying all invariants.
Furthermore, when \attr\ is applied on a justified $v$, it preserves $H(v)$ but it replaces $v$'s direct justification by one with a strictly higher justification level. As the proof below shows, this ensures that no new cycles are created through $v$ so we can guarantee that all remaining cycles still have the correct winner. So, cycles can only be created by justifying an unjustified node.    

\begin{lemma}\label{theo:safe} An executable operation \justify{} returns a safe justification.
\end{lemma}

\begin{proof}
Assume \justify{} is executable, $J'=\justify$ and let $\alpha$ be the player that wins $v$ by $dj$.  First, we prove that $J'$ is also a winning justification, i.e., that $J'$ is weakly winning and that the winner of every infinite path in $J'$ is the hypothetical winner $H(w)$ of the nodes $w$ on the path. 

The operations applied to obtain $J'$ are the  ones that have been analysed in Proposition \ref{prop:operations} and for which it was proven that they preserve  weakly winning.
Note that, in case $H(v) = \oppon$, the intermediate justification $J[v: \emptyset, \defwin(v) \mid v \in \reaching{J}{v}]$ removes all incoming edges of $v$. Hence, $J'$ is weakly winning and all nodes $v, w$ connected in $J$ have  $H'(v) = H'(w)$ (*).
If no edge in $dj$ belongs to a cycle, then every infinite path $\play{v_1, v_2, \dots}$ in $J'$ has an infinite tail in $J$ starting in $w \neq v$ which is, since $J$ is winning, won by $H(w)$. By (*), this path is won by $H(v_1) = H(w)$ and $J'$ is winning. 
  
  If $J'$ has cycles through edges in $dj$, then, by (i) of Proposition~\ref{prop:operations}, $H(v)$ must be $\alpha$ and we are in case (iii)
  of Proposition \ref{prop:operations}.
  We analyse the nodes
  $n$ on such a cycle. By safety of $J$, $Pr(n) \leq jl_J(n)$; as $n$
  reaches $v$ in $J$, $jl_J(n) \leq jl_J(v)$. If $v$ is
  unjustified in $J$ then $jl_J(v) = Pr(v) \geq Pr(n)$, hence $Pr(v)$
  is the highest priority on the cycle and $H(v)$ wins the cycle. If
  $v$ is justified in $J$ and $(v,w) \in dj$ is on the new cycle, then
  $jl_J(w) \geq jl_J(dj) > jl_J(v)$ (Precondition 2 of \attr). But $w$
  reaches $v$ so $jl_J(w) \leq jl_J(v)$ , which is a contradiction.

Next, we prove that J' is a safe justification
(Definition~\ref{def:extdom}). (i)  We just proved that $J'$ is a winning justification. 
  (ii) For all unjustified nodes $v$ of $J'$, it holds that $H(v)=\defwin(v)$, its default winner. Indeed, $J$ has this property and whenever the direct justification of a node $w$ is removed, $H'(w)$ is set to $\defwin(w)$. 

(iii) We need to prove  that for all nodes $w$,  it holds that $\jl_{J'}(w) \geq Pr(w)$. We  distinguish between the two  cases of \justify.

(a) Assume $H(v) = \alpha = H'(v)$ and $J'=J[v:dj,H(v)]$ and let $w$ be an arbitrary node of $V$. If  $w$  cannot reach $v$ in $J'$, the parameters that $w$ reaches in $J$ and $J'$ are the same and it follows that $\jl_{J'}(w)=\jl_J(w)\geq Pr(w)$. So, (iii) holds for $w$.  Otherwise, if $w$  reaches $v$ in $J'$, then $w$ reaches $v$ in $J$ and  any parameter $x$ that $w$ reaches in $J'$ is a parameter that $w$ reaches in $J$ or one that an element of $dj$ reaches in $J$. It follows that $\jl_{J'}(w)$ is
at least 
the minimum of $\jl_J(w)$ and $\jl_J(dj)$. As $w$ reaches $v$ in $J$,  $\jl_J(w)\leq\jl_J(v)$. Also,  by  Precondition 2 of \attr, $\jl_J(v)\leq\jl_J(dj)$. It follows that $\jl_{J'}(w)\geq\jl_J(w)\geq Pr(w)$. Thus, (iii) holds for $w$. 

(b)  Assume  $H'(v) \neq H(v) = \oppon$ and $	J'=J[w:\emptyset,\defwin(w)\mid w\in \reaching{J}{v}][v:dj,\alpha]$  then for nodes $w$ that cannot reach $v$ in $J$, $\Ext_{J'}(w) = \Ext_J(w)$  hence $\jl_{J'}(w) =  \jl_J(w) \geq Pr(w)$ and (iii) holds for $w$. All nodes $w\neq v$ that can reach $v$ in $J$ are reset, hence  $\jl_{J'}(w) = Pr(w)$ and (iii) holds. As for $v$, by construction $\jl_{J'}(v) = \jl_J(dj) \geq \jl_J(v)$; also $\jl_J(v) \geq Pr(v)$ hence (iii) also
holds. \qed
\end{proof}

\begin{lemma}\label{theo:monotonic} Let $J$ be a safe justification for a
  parametrized parity game. Unless $J$ defines the parametrized parity
  game $PG_\emptyset=\PG$,  there exists a node $v$ justifiable
  with a direct justification $dj$, i.e., such that \justify{} is executable.
\end{lemma}

\begin{proof}
  If $J$ defines the parametrized parity game
  $PG_\emptyset$ then all nodes are justified and $J$ is a solution for the original $\PG$. Otherwise let $p$ be the minimal priority of all
  unjustified nodes, and $v$ an arbitrary unjustified node of priority $p$ and let its owner be 
  $\alpha$.  Then either $v$ has an outgoing edge
  $(v,w)$ to a node $w$ with $H(w) = \alpha$, thus a winning direct
  justification for $\alpha$, or all outgoing edges are to nodes
  $w$ for which $H(w) = \oppon$, thus $v$ has a winning direct justification for
  $\oppon$.  In both cases, this direct justification $dj$ has a
  justification level larger or equal to $p$ since no parameter
  with a smaller priority exist, so \justify{} is executable.
  \qed
\end{proof}

To show progress and termination, we need an order over
justifications.

\begin{definition}[Justification size and order over justifications]
  Let $1,\ldots,n$ be the range of the priority function of a parity
  game $PG$ ( $\pinf > n$) and $J$ a winning justification for a parametrized parity
  game extending $PG$.
  The size of $J$, $s(J)$ is the tuple $(s_{\pinf}(J),s_n(J), \ldots s_1(J))$ where for $i \in \{1,\dots,n,\dominion\}$, $s_i(J)$ is the number of
  justified nodes with justification level $i$. 

  The order over justifications is the lexicographic order over their
  size: with $i$ the highest index such that $s_i(J) \neq s_i(J')$,
  we have $J >_s J'$ iff $s_i(J) > s_i(J')$. 
\end{definition}
The order over justifications is a total order which is bounded as $\Sigma_i s_i(J) \leq |V|$.

\begin{figure} [tb]
\begin{center}
\includegraphics[width=0.6\textwidth]{images/game_para_e_e.pdf}
\end{center}
\begin{center}
\includegraphics[width=0.6\textwidth]{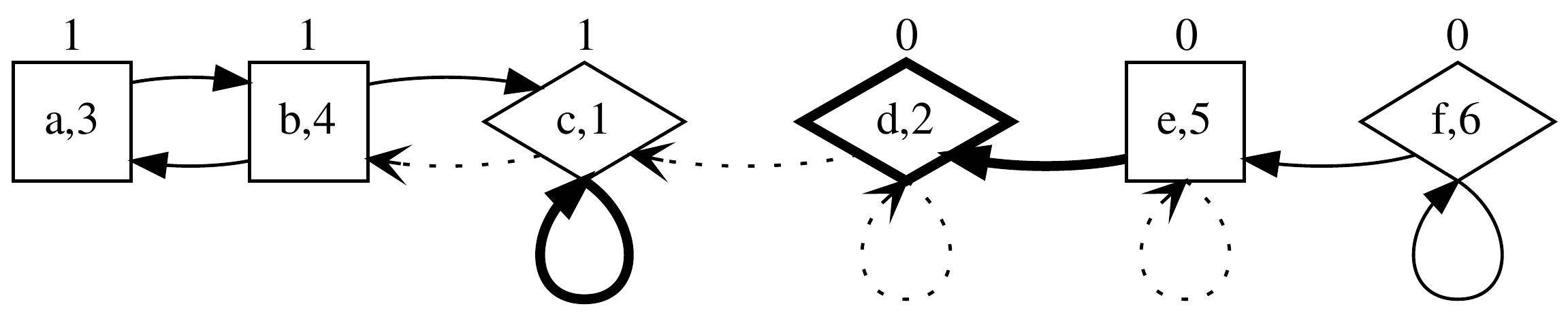}
\end{center}
\caption{Above,
	 in solid line the edges of the justification graph of
  the winning but unsafe justification of Figure~\ref{ex:game_para}
  and below the result of justifying node $a$, a non-winning justification.\label{fig:just}}
\end{figure}

\begin{example}\label{ex:second}  Let us revisit Example~\ref{ex:first}. The winning justification $J$
  of Figure~\ref{ex:game_para} is shown at the top of
  Figure~\ref{fig:just}. For the justified nodes of $J$, we have
  $\jl(b)=3$, $\jl(c)=\dominion$, $\jl(e)=2$ and $\jl(f)=2$. The justification is not safe as, e.g., $\jl(b)=3 <
  Pr(b)=4$. Both unjustified nodes $a$ and $d$ have a winning direct justification, the direct justification $\{(a,b)\}$ wins $a$ for player 1 and the direct justification
  $\{(d,c)\}$ wins $d$ for 1. The figure at the bottom shows the justification resulting from inserting the direct justification winning $a$. There is
  now an infinite path $\play{a,b,a,b,\ldots}$ won by Even but with nodes with hypothetical winner Odd. The justification $\attr(J,a,\set{(a,b)})$ is not winning. This shows that condition (iii) of safety of $J$  is a necessary
  precondition for maintaining the desired invariants.
\end{example}

\begin{lemma}\label{theo:size}
  Let $J$ be a safe justification with size $s_J$, $v$ a node
  justifiable with $dj$ and $J'=\justify{}$ a justification with size
  $s_{J'}$. Then $s_{J'} > s_J$.
\end{lemma}
\begin{proof}
  In case $v$ is unjustified in $J$ and is assigned a $dj$ that wins
  $v$ for $H(v)$, $v$ is not counted for the size of $J$ but is
  counted for the size of $J'$. Moreover, other nodes keep their
  justification level (if they cannot reach $v$ in $J$) or may
  increase their justification level (if they can reach $v$ in
  $J$). In any case, $s_{J'} > s_J$. 
	
  In case $v$ is justified in $J$ and is assigned a $dj$ that wins $v$
  for $H(v)$, then $\jl_J(dj) > \jl_J(v)$, so
  $\jl_J'(v) > \jl_J(v)$.  Other nodes keep their justification level
  or, if they reach $v$, may increase their justification
  level. Again, $s_{J'} > s_J$.  
	
  Finally, the case where $dj$ wins $v$ for the opponent of
  $H(v)$. Nodes can be reset; these nodes $w$ have $\jl_J(w)\leq
  Pr(v)$. As a node cannot have a winning direct justification for
  both players, $v$ is unjustified in $J$. Hence, by precondition (2)
  of \attr, $\jl_J(dj)\geq Pr(v)$. In fact, it holds that $\jl_J(dj)>
  Pr(v)$. Indeed, if some $w\in dj$ would have a path to a parameter
  of $v$'s priority, that path would be won by $\defwin(v)=H(v)$ while
  $H(w)$ is its opponent. Thus, the highest index $i$ where $s_i$
  changes is $\jl_J(dj)$, and $s_i$ increases.  Hence, $s_{J'}>s_J$.
  \qed
\end{proof}

\begin{theorem}\label{theo:basic}
  Any iteration of \attr{} steps from a safe justification, in
  particular from $(V, \emptyset,\defwin)$, with $\defwin$ the default hypothesis, eventually solves $\PG$.
\end{theorem}
\begin{proof}
  By induction: Let $\PG=(V,E,\Owns,Pr)$ be a parity game. Clearly, the  empty justification $J^0=(V,\emptyset,\defwin)$ is a safe justification. This is the base case. 

  Induction step: Let $J^i$ be the safe justification after $i$ successful \attr\ steps and assume that $J^i=(V,D^i,H^i)$ 
  contains an unjustified node.  
  By Lemma~\ref{theo:monotonic}, there exists a pair $v$ and $dj$ such that $v$ is justifiable with $dj$. For {\em any} pair $v$ and $dj$ such that $\attr(J^i,v,dj)$ is executable, let $J^{i+1}=\attr(J^i,v,dj)$. By
  Lemma~\ref{theo:safe}, $J^{i+1}$ is a safe justification. By
  Lemma~\ref{theo:size}, there is a strict increase in size, i.e.,
  $s(J^{i+1}) > s(J^i)$. 

  Since the number of different sizes is bounded, this eventually produces a safe $J^k=(V,D^k,H^k)$ without unjustified nodes. 
  The parametrized parity game $\PG_{P_{J^k}}$ determined by $J^k$ is $\PG$. Hence, $H^k$ is the winning function of $\PG$, and $J^k$ comprises winning strategies for both players. \qed
\end{proof}

Theorem~\ref{theo:basic} gives a basic algorithm  to solve parity games. The algorithm has three features: it is (1) simple, (2) {\em nondeterministic}, and (3) in successive steps it may arbitrarily switch between different priority levels.
Hence, by imposing different strategies, different instantiations of the algorithm are obtained.

Existing algorithms differ in the order in which they (implicitly) justify
nodes. In the next section we simulate such algorithms by
different strategies for selecting nodes to be justified.  Another
difference between algorithms is in computing the set $R$ of nodes that is
reset when $dj$ wins $v$ for the opponent of $H(v)$. Some algorithms
reset more nodes; the largest reset set for which the proofs in this
paper remain valid is $\{w \in V \mid \jl(w) < \jl(dj)\}$. To the best
of our knowledge, the only algorithms that reset as few nodes as
\justify{} are the ones we presented in
\cite{DBLP:conf/vmcai/LapauwBD20}. As the experiments presented there
show, the work saved across iterations by using justifications
results in better performance.
\newcommand{\zielonka}{{\operation{Zielonka}}}
\newcommand{\promote}{{\operation{Promote}}}
\newcommand{\isclosed}{{\operation{Closed}}}
\newcommand{\fixpoint}{{\operation{Fixpoint}}}

\section{A reformulation of three existing algorithms\label{sec:recursive}}

In this section, by ordering justification steps, we obtain basic
versions of different algorithms known from the literature. 
In our versions, we represent the parity game $G$ as $(V,E,O,Pr)$ and the justification J as $(V,D,H)$.
All algorithms start with the safe empty
justification $(V, \emptyset, \defwin)$. The recursive algorithms
operate on a subgame $SG$ determined by a set of nodes $V_{SG}$. This
subgame determines the selection of \justify{} steps that are
performed on $G$. For convenience of presentation, $G$ is considered
as a global constant.

\subsubsection{Nested fixpoint iteration~\cite{DBLP:journals/corr/BruseFL14,DBLP:journals/corr/abs-1909-07659,DBLP:conf/vmcai/LapauwBD20}} is one of the earliest algorithms
able to solve parity games. In Algorithm~\ref{alg:fpi}, we show a
basic form that makes use of our \justify{} action. It starts from the
initial justification $(V,\emptyset,H_d)$. Iteratively, it determines
the lowest priority $p$ over all unjustified nodes, it selects a node
$v$ of this priority and justifies it. Recall from the proof of
Lemma~\ref{theo:monotonic}, that all unjustified nodes of this priority are
justifiable. Eventually, all nodes are justified and a solution is
obtained. For more background on nested fixpoint algorithms and the
effect of justifications on the performance, we refer to our work in
\cite{DBLP:conf/vmcai/LapauwBD20}.

A feature of nested fixpoint iteration is that it solves a parity game
{\em bottom up}. It may take many iterations before it uncovers that the current hypothesis of some high priority unjustified node $v$ is, in fact,
wrong and so that playing to $v$ is a bad strategy for
$\alpha$. The next algorithms are {\em top down}, they start out from
nodes with the highest priority.

\begin{figure}[t]
{\begin{minipage}{0.46\textwidth}
\begin{algorithm}[H]
\algsetup
\Fn{$\fixpoint(G)$}{
    $J \is (V,\emptyset,\defwin)$ the initial safe justification\\
  \While{$J$ has unjustified nodes} {
    $p \is min\set{Pr(v) \given v \text{ is unjustified} }$ \\
    $v \is $ an unjustified node with $Pr(v) = p$ \\
    $\djn \is $ a winning direct justification for $v$ under $H$\\
    $J \is \attr(J, v, \djn)$\\
  }
  \Return{$J$}
}
\caption{A fixpoint algorithm for justifying nodes\label{alg:fpi}}
\end{algorithm}
\end{minipage}~
\begin{minipage}{0.54\textwidth}
\begin{algorithm}[H]
\algsetup
\Input{A parity game $G$}
$J \is \zielonka((V,\emptyset,\defwin), V)$ \label{line:zlk_init}\\
\Fn{$\zielonka(J, V_{SG})$}{
	$p \is max\set{Pr(v) \given v \in V_{SG}}$ \\
	$\alpha \is \allied p$\\
	\While{$true$}{
    	\While {$\exists v \in V_{SG}, dj: v$ is unjustified, $v$ is justifiable with $\djn$ for $\alpha$ with $jl(\djn)\geq p$  \label{line:zlk:greedy}}{
			$J \is \attr(J, v, \djn)$ \label{line:zlk_just}\\
		}
        $V_{SSG} \is \{ v \in V_{SG}| Pr(v)<p,$\\
        ~~~~~~~~~~~~ $v$ is unjustified\}\\
		\lIf{$V_{SSG} = \emptyset$} {\Return{$J$}}
		$J \is \zielonka(J, V_{SSG})$\\
			 \While {$\exists v \in V_{SG}, dj: v$ is unjustified, $v$ is justifiable with $\djn$ for $\oppon$ with $jl(\djn)\geq p+1$  \label{line:zlk:oppon:attr}}{
				$J \is \attr(J, v, \djn)$\\
			}
	}
}
\caption{A \attr{} variant of Zielonka's algorithm.}\label{alg:ZielJust}
\end{algorithm}
\end{minipage}}
\end{figure}

\subsubsection{Zielonka's algorithm~\cite{DBLP:journals/tcs/Zielonka98},}
one of the oldest algorithms, is recursive and starts with a greedy
computation of a set of nodes, called {\em attracted} nodes, in which the winner $\alpha$ of the top priority $p$ has a strategy to force playing to nodes of top priority $p$.  
In our reconstruction, Algorithm~\ref{alg:ZielJust}, attracting nodes is simulated at Line~\ref{line:zlk:greedy} by repeatedly  justifying  nodes $v$ with a direct justification that wins $v$ for $\alpha$ and has a justification level $\geq p$. Observe that the while test ensures that the preconditions of \justify{} on the justification level of $v$ are satisfied. Also, every node \ignore{in $V'$}can be justified at most once.

The procedure is called with a set $V_{SG}$ of nodes of maximal level
$p$ that cannot be attracted by levels $>p$. It follows that the
subgraph determined by 
$V_{SG}$ contains for each of its nodes an outgoing edge (otherwise the opponent of
the owner of the node would have attracted the node at a level $>p$) , hence
this subgraph determines a parity game.
The main loop invariants are that (1) the justification $J$ is safe; (2) the justification level of all justified nodes is $\geq p$ and (3) $\oppon$ has no direct justifications of justification level $> p$ to win an unjustified node in $V_{SG}$.  
The initial justification is safe and it remains so as every $\attr$
call satisfies the preconditions.

After the attraction loop at Line~\ref{line:zlk:greedy}, no more unjustified nodes of $V_{SG}$  can be attracted to level $p$ for player $\alpha$. Then, the set of $V_{SSG}$ of unjustified nodes of priority $<p$ is determined. If this set is empty, then by Lemma~\ref{theo:monotonic} all unjustified nodes of priority $p$ are justifiable with a direct justification $dj$ with $jl(dj)\geq p$, hence they would be attracted to some level $\geq p$ which is impossible.  Thus,  there are no unjustified nodes of priority $p$. In this case,  the returned justification $J$ justifies all elements of $V_{SG}$.   Else,  $V_{SSG}$ is passed in a recursive call to justify all its nodes. Upon return, if $\oppon$ was winning some nodes in $V_{SSG}$, their justification level will be $\geq p+1$. Now it is possible that some unjustified nodes of priority $p$ can be won by $\oppon$ and this may be the start of a cascade of resets and attractions for $\oppon$. The purpose of  Line~\ref{line:zlk:oppon:attr} is to attract nodes of $V_{SG}$ for $\oppon$.
Note that \justify{} resets all nodes that depend on nodes that switch to $\oppon$.
When the justification returned by the recursive call shows that $\alpha$ wins
all nodes of $V_{SSG}$, the yet unjustified nodes of $V_{SG}$ are of priority $p$, are justifiable  by Lemma~\ref{theo:monotonic} and can be won only by $\alpha$. So, at the next iteration, the call to $Attr_\alpha$ will  justify all of them for $\alpha$ and $V_{SSG}$ will be empty.  Eventually the initial call of Line~\ref{line:zlk_init} finishes with a safe justification in which all nodes are justified thus solving the game $G$. 
    
Whereas fixpoint
iteration first justifies low priority nodes resulting in low justification levels, Zielonka's algorithm
first justifies nodes attracted to the highest priority.  Compared to
fixpoint iteration, this results in large improvements in
justification size which might explain its better performance. However, Zielonka's
algorithm still disregards certain opportunities for increasing
justification size as it proceeds by priority level, only returning to
level $p$ when all sub-problems at level $<p$ are completely
solved. Indeed, some nodes computed at a low level $i<\!\!<p$ may have
a very high justification level, even $\dominion$ and might be useful
to revise false hypotheses at high levels, saving much work, but this
is not exploited.  The next algorithm, priority promotion, overcomes
this limitation.

\subsubsection{Priority promotion~\cite{DBLP:conf/cav/BenerecettiDM16,DBLP:conf/hvc/BenerecettiDM16,DBLP:journals/iandc/BenerecettiDM18}}
follows the strategy of Zielonka's algorithm except that, {when it
  detects that all nodes for priority $p$ are justified, it does not
  make a recursive call but returns the set of nodes attracted to
  priority $p$ nodes as a set $R_p$ to a previous level $q$.}
There $R_p$ is added to the attraction set at level $q$ and the attraction process is restarted. In the terminology of  \cite{DBLP:conf/cav/BenerecettiDM16},  the  set  $R_p$ is a  {\em  closed  $p$-region} that is {\em promoted} to level $q$.  A {\em
	closed $p$-region} of $V_{SG}$, with maximal priority $p$, is a subset $R_p \subseteq V_{SG}$  that includes all
nodes of $V_{SG}$ with priority $p$ and for which $\alpha = \allied{p}$ has
a strategy winning all infinite plays in $R_p$ and for which $\oppon$ cannot escape from $R_p$ unless to nodes of higher $q$-regions won by $\alpha$. 
We call the latter nodes the {\em escape nodes} from $R_p$ denote the set of them as $Escape(R_p)$. 
The level to which $R_p$ is promoted is the lowest $q$-region that contains an escape node from $R_p$. 
It is easy to show that $q$ is a lower bound of the justification level of $R_p$. In absence of escape nodes, $R_p$ is promoted to $\dominion$.  

\begin{figure}[tb]
	\begin{minipage}{0.45\textwidth}
		\begin{algorithm}[H]
			\algsetup
			\Input{A parity game $G$ }
			$J \is (V,\emptyset,\defwin)$\\
			\While{$\exists v \in V_G: v$ is unjustified}{
				$R_{\dominion} \is \set{ v \given  jl(v) = \dominion}$ \\
				$V_{SG} \is V \setminus R_{\dominion}$ \\
				$(J, \_,\_) \is \promote(V_{SG},J)$\\
				\While {$\exists v \in V_{SG}, dj: v$ is justifiable with $\djn$ and $jl(\djn) = \dominion$}{
					$J \is \attr(J, v, \djn)$ \label{line:attr_solve2}\\
				}
			}
			\caption{A variant of priority promotion using \attr. \label{alg:Mmprom}}
		\end{algorithm}
	\end{minipage}~
	\begin{minipage}{0.55\textwidth}
		\begin{algorithm}[H]
			\algsetup
			\Fn{$\promote(V_{SG},J)$}{
				$p \is max\set{Pr(v) \given v \in V_{SG}}$ \\
				$\alpha \is \allied{p}$ \\
				\While{$true$}{  
					\While {$\exists v \in V_{SG}, dj: v$ is unjustified or $jl(v)<p$, $v$ is justifiable with $\djn$ for $\alpha$ with $jl(\djn)\geq p$  \label{line:greedy}}{
						$J \is \attr(J, v, \djn)$ \\}
					$R_p \is \set{v \in V_{SG} \given jl(v) \geq p}$ \\ 
					\If{$\isclosed(R_p, V_{SG})$}{
						$l \is min\{q| R_q$ contains an escape node of $R_p\}$ \label{line:escapelevel}\\
						\Return{$(J,R_p,l)$}
					}
					$V_{SSG} \is V_{SG} \setminus R_p$ \\
					$(J,R_{p'},l) \is \promote(V_{SSG}, J)$\\
					\If{$l > p$\label{line:backjump}}{\Return{$(J,R_{p'},l)$}}\label{line:skip}
				}
			}
		\end{algorithm}
	\end{minipage}
\end{figure}

Our variant of priority promotion (PPJ) is in Algorithm~\ref{alg:Mmprom}. 
Whereas \zielonka{} returned a complete solution $J$ on $V_{SG}$, \promote{} returns only a partial $J$ on $V_{SG}$; some nodes of $V_{SG}$ may have an unfinished justification ($jl(v)<\dominion$). 
To deal with this, \promote{} is iterated in a while loop that continues as long as there are
unjustified nodes. 
Upon return of \promote{}, all nodes attracted to the returned $\dominion$-region are justified. 
In the next iteration, all nodes with justification level $\dominion$ are removed from the game, permanently.
Note that when promoting to some $q$-region, justified nodes of justification level $<q$ can remain.
A substantial gain can be obtained compared to the original priority promotion algorithm which does not maintain justifications and loses all work stored in $J$.

By invariant, the function \promote{} is called with a set of nodes $V_{SG}$ that cannot be justified with a direct justification of level larger than the maximal priority $p$. 
The function starts its main loop by attracting nodes for level $p$. 
The attraction process is identical to Zielonka's algorithm except that leftover justified nodes $v$ with $jl(v)<p$ may be rejustified. 
As before, the safety of $J$ is preserved. 
Then $R_p$ consists of elements of $V_{SG}$ with justification level $\geq p$. 
It is tested (\isclosed{}) whether $R_p$ is a closed $p$-region. 
This is provably the case if all nodes of priority $p$ are justified.  
If so, $J$, $R_p$ and its minimal escape level are returned. 
If not, the game proceeds as in Zielonka's algorithm and the
game is solved for the nodes not in $R_p$ which have strictly lower justification level. 
Sooner or later, a closed region will be obtained. 
{Indeed, at some point, a subgame is entered in which all nodes have
the same priority $p$. All nodes are justifiable
(Lemma~\ref{theo:monotonic}) and the resulting region is closed.}
Upon return from the recursive call, it is
checked whether the returned region ($R_{p'}$) promotes to the current
level $p$.
If not, the function exits as well (Line~\ref{line:skip}).
Otherwise a new iteration starts with attracting nodes of justification  level $p$ for $\alpha$. Note that contrary to Zielonka's algorithm, there is no attraction step for $\oppon$:
attracting for $\oppon$ at $p$ is the same as attracting for $\alpha' = \oppon$ at $p' = p+1$.
\subsubsection{Discussion}
Our versions of Zielonka's algorithm and priority promotion use the
justification level to decide which nodes to attract. While
maintaining justification levels can be costly, in these algorithms,
it can be replaced by selecting nodes that are ``forced to play'' to
a particular set of nodes (or to an already attracted node). In the
first attraction loop of \zielonka, the set is initialised with
all nodes of priority $p$, in the second attraction loop, with the
nodes won by $\oppon$; In \promote, the initial set consists also of
the nodes of priority $p$.

Observe that the recursive algorithms implement a strategy to reach as soon
as possible the justification level $\dominion$ for a group of nodes
(the nodes won by the opponent in the outer call of \zielonka, the
return of a closed region ---for any of the players--- to the outer
level in \promote). When achieved, a large jump in justification size
follows. This may explain why these algorithms outperform fixpoint
iteration.

Comparing our priority promotion algorithm (PPJ) to other variants, we
see a large overlap with region recovery (RR)~\cite{DBLP:conf/hvc/BenerecettiDM16} both algorithms avoid
resetting nodes of lower regions. However, RR always resets the full
region, while PPJ can reset only a part of a region, hence can save more
previous work.
Conversely, PPJ eagerly resets nodes while RR only validates the
 regions before use, so it can 
recover a region when the reset escape node is easily re-attracted.
The equivalent justification of such a state is winning but unsafe, thus unreachable by applying $\justify{}$.
However, one likely can define a variant of $\justify{}$ that can  reconstruct RR.
Delayed priority promotion~\cite{DBLP:journals/iandc/BenerecettiDM18} is another variant which prioritises the best
promotion over the first promotion and, likely, can be directly reconstructed.

Tangle learning~\cite{DBLP:conf/cav/Dijk18} is another state of
the art algorithm that we have studied. Space restrictions disallow us
to go in details. We refer to \cite{DBLP:conf/vmcai/LapauwBD20} for a
version of tangle learning with justifications. For a more formal
analysis, we refer to~\cite{LapauwPhD}). Interestingly, the updates of the
justification in the nodes of a tangle cannot be modelled with a
sequence of safe \justify{} steps. One needs an alternative with a
precondition on the set of nodes in a tangle. Similarly as for \justify{},
it is proven in~\cite{LapauwPhD} that the resulting justification is
safe and larger than the initial one.

Justification are not only a way to explicitly model (evolving)
winning strategies, they can also speed up algorithms. We have
implemented justification variants of the nested fixpoint algorithm,
Zielonka's algorithm, priority promotion, and tangle learning. For the
experimental results we refer to \cite{DBLP:conf/vmcai/LapauwBD20,LapauwPhD}.

Note that the data structure used to implement the justification graph
matters. Following an idea of Benerecetti et al.\cite{DBLP:conf/cav/BenerecettiDM16}, our
implementations use a single field to represent the direct
justification of a node; it holds either a single node, or
$\mathit{null}$ to represent the set of all outgoing nodes.
To compute the reset set $R$ of a node, we found two efficient methods to encode the graph $J$: (i) iterate over all incoming nodes in $E$ and test if their justification contains $v$, (ii) store for every node a hash set of every dependent node.
On average, the first approach is better, while the second is more efficient for sparse graphs but worse for dense graphs.

\section{Conclusion}\label{sec:conclusion}
This paper explored the use of justifications in parity game
solving.
First, we generalized parity games by adding parameter
nodes. When a play reaches a parameter it stops in favour of one
player. Next, we introduced justifications and proved that a winning
justification contains the solution of the parametrized parity game.
Then, we introduced safe justifications and a \attr\ operation and
proved that a parity game can be solved by a sequence of \attr\ steps.
A \attr\ operation can be applied on a node satisfying its
preconditions, it assigns a winning direct justification to the node,
resets ---if needed--- other nodes as parameters, preserves safety
of the justification, and ensures the progress of the solving process.

To illustrate the power of \attr, we reconstructed three algorithms:
nested fixpoint iteration, Zielonka's algorithm and priority promotion
by ordering applicable $\attr{}$ operations differently.  Nested
fixpoint induction prefers operations on nodes with the lowest
priorities; Zielonka's algorithm starts on nodes with the maximal
priority and recursively descends; priority promotion improves upon
Zielonka with an early exit on detection of a closed region (a solved
subgame).

A distinguishing feature of a justification based algorithm is that it
makes active use of the partial strategies of both players.
While other algorithms, such as region recovery and tangle learning, use the constructed partial strategies while solving the parity game, we do not consider them justification based algorithms. 
For region recovery, the generated states are not always weakly winning, while tangle learning applies the partial strategies for different purposes.
As shown in \cite{DBLP:conf/vmcai/LapauwBD20} where justifications improve tangle learning, combining different techniques can further improve parity game algorithms.

Interesting future research includes: (i) exploring the
possible role of justifications in the quasi-polynomial algorithm of
Parys~\cite{DBLP:conf/mfcs/Parys19}, (ii) analysing the similarity
between small progress measures
algorithms~\cite{DBLP:conf/spin/FearnleyJS0W17,DBLP:conf/stacs/Jurdzinski00}
and justification level, (iii) analysing whether the increase in
justification size is a useful guide for selecting the most
promising justifiable nodes, (iv) proving the worst-case time
complexity by analysing the length of the longest path in the lattice
of justification states where states are connected by \justify{} steps.
\bibliography{biblio}
\bibliographystyle{splncs04}
\end{document}